\documentclass[11pt,a4paper]{article}
\usepackage[utf8]{inputenc}
\usepackage[T1]{fontenc}

\usepackage{amsmath}
\usepackage{amsfonts}
\usepackage{amssymb}
\usepackage{amsthm}
\usepackage[
left=3.5cm, right=3.5cm, top=3.2cm, bottom=3.2cm
]{geometry}
\usepackage{authblk}
\usepackage{multicol}
\usepackage{stmaryrd}
\usepackage{mathrsfs}
\usepackage{hyperref}

\newtheorem{theorem}{Theorem}
\newtheorem{lemma}{Lemma}
\newtheorem{proposition}{Proposition}

\newtheorem{definition}{Definition}
	\theoremstyle{definition}
\newtheorem{example}{Example}
	\theoremstyle{remark}

\newcommand{\rKAt}{\mathbf{KA}\vec{\mathbf{t}}}
\newcommand{\g}[1]{\mathscr{#1}}

\newcommand{\at}[1]{\mathsf{#1}}
\newcommand{\eqt}{\approx}
\newcommand{\val}[1]{{\normalfont \texttt{[}#1\texttt{]}}}
\newcommand{\vval}[1]{\!\text{\normalfont\texttt{[\!\![}} #1 \text{\normalfont\texttt{]\!\!]}}}
\newcommand{\tra}{Tr}  

\title{One-sorted Program Algebras}

\author{Igor Sedl\'{a}r}
\author{Johann J.~Wannenburg}
\affil{Czech Academy of Sciences, Institute of Computer Science}

\date{April 27, 2022}

\begin{document}
\maketitle

\begin{abstract}
Kleene algebra with tests, KAT, provides a simple two-sorted algebraic framework for verifying properties of propositional while programs. Kleene algebra with domain, KAD, is a one-sorted alternative to KAT. The equational theory of KAT embeds into KAD, but KAD lacks some natural properties of KAT. For instance, not each Kleene algebra expands to a KAD, and the subalgebra of tests in each KAD is forced to be the maximal Boolean subalgebra of the negative cone. In this paper we propose a generalization of KAD that avoids these features while still embedding the equational theory of KAT. We show that several natural properties of the domain operator of KAD can be added to the generalized framework without affecting the results. We consider a variant of the framework where test complementation is defined using a residual of the Kleene algebra multiplication.
\end{abstract}

\section{Introduction}

Kleene algebra with tests \cite{Kozen1997}, KAT, is a simple algebraic framework for verifying properties of propositional while programs. KAT subsumes propositional Hoare logic \cite{Kozen2000} and it has been applied in a number of verification tasks. KAT has computationally attractive fragments \cite{SmolkaEtAl2020} and its extensions have been applied beyond while programs, for instance in network programming languages \cite{AndersonEtAl2014}. 

KAT is two-sorted, featuring a Boolean algebra of tests embedded into a Kleene algebra of programs. For various reasons, a one-sorted alternative to KAT may be desirable. One such alternative is Kleene algebra with domain, KAD, in its ``internal'' formulation \cite{DesharnaisStruth2011}. The idea of KAD is to expand Kleene algebra with a domain operator $d$ such that the set of images of elements of the algebra under $d$ forms a Boolean algebra. Hence, one obtains a Boolean algebra of tests in a one-sorted setting. Consequently, the equational theory of KAT embeds into the equational theory of KAD.

However, the properties of the domain operator in KAD lead to some curious features of the framework. For instance, the algebra of tests is always the maximal Boolean subalgebra of the negative cone in the underlying Kleene algebra. In a sense, then, every ``proposition'' is considered a test, contrary to some of the intuitions expressed in \cite{Kozen1997}. Moreover, not each Kleene algebra can be expanded with a domain operator, not even every finite one. This is in contrast with the fact that each KA is expands to a KAT.

In this paper we generalize KAD to a framework we call \emph{one-sorted Kleene algebra with tests}, KAt. Our framework has the desired features of KAD: every KAt contains a Boolean subalgebra of tests and the equational theory of KAT embeds into the equational theory of KAt. In addition, every Kleene algebra expands into a KAt, and the subalgebra of tests in KAt is not necessarily the maximal Boolean subalgebra of the negative cone. We also consider various extensions of KAt with axioms known from KAD to show which properties of the domain operator are still consistent with the desired features of KAt. In addition, we consider a variant of the KAt framework where test complementation is defined using a residual of Kleene algebra multiplication. It can be shown that algebras in this framework form a variety. 

The paper is organized as follows. In Sections \ref{sec:KAT} and \ref{sec:KAD} we recall Kleene algebra with tests and Kleene algebra with domain, respectively. In Section \ref{sec:KAt} we introduce one-sorted Kleene algebra with domain. We show that every Kleene algebra with domain is a one-sorted Kleene algebra with tests, that the test algebra in every one-sorted KAT is a Boolean algebra, and that every Kleene algebra expands into a one-sorted Kleene algebra with tests. In Section \ref{sec:ext} we consider some extensions of one-sorted KAT that retain the good properties. In Section \ref{sec:embedding} we show that the equational theory of KAT embeds into the equational theory of  one-sorted KAT and a number of its extensions. Section \ref{sec:residual} briefly considers a residuated variant of one-sorted KAT that forms a variety.

\section{Kleene algebra with tests}\label{sec:KAT}

We assume that the reader is familiar with the notion of an \emph{idempotent semiring}.  The definition of Kleene algebra is taken from \cite{Kozen1994}. The reader is referred to \cite{Kozen1997} for more details on Kleene algebra with tests.

\begin{definition}
A \emph{Kleene algebra} is an idempotent semiring $(K, \cdot, +, 1, 0)$ expanded with an operation $^{*} : K \to K$ such that
\begin{gather*}
1 + x x^{*} \leq x^{*}\\
1 + x^{*}x \leq x^{*}\\
y + xz \leq z \implies x^{*}y \leq z\\
y + zx \leq z \implies y x^{*} \leq z
\end{gather*}

A \emph{Kleene algebra with tests} is $(K, B, \cdot, +, ^{*}, 1, 0, \, \bar{ } \, )$ such that
\begin{itemize}
\item $(K, \cdot, +, ^{*}, 1, 0)$ is a Kleene algebra,
\item $B \subseteq K$, and
\item $(B, \cdot, +, \, \bar{ }\, , 1, 0) $ is a Boolean algebra.
\end{itemize}
\end{definition}
A standard example of a Kleene algebra is a relational Kleene algebra where $K$ is a set of binary relations over some set $S$, $\cdot$ is relational composition, $+$ is set union, $^{*}$ is reflexive transitive closure, $1$ is identity on $S$ and $0$ is the empty set; another standard example is the Kleene algebra of regular languages over some finite alphabet. 

Every KA is a KAT; take $B = \{ 0, 1 \}$ and define $\bar{0} = 1$, $\bar{1} = 0$ and $\bar{x} = x$ for $x \notin \{ 1, 0 \}$. The class of Kleene algebra with tests is denoted as $\mathbf{KAT}$. A standard example of a KAT is a relational KA expanded with a Boolean subalgebra of the \emph{negative cone}, i.e.~the elements $x \leq 1$, which in the relational case are subsets of the identity relation. The class of relational Kleene algebras with tests will be denoted as $\mathbf{rKAT}$.

The equational theory of $\mathbf{KAT}$ is defined in terms of a two-sorted language. Let $\mathsf{B} = \{ \mathsf{b}_n \mid n \in \omega \}$ and $\mathsf{P} = \{ \mathsf{p}_{n} \mid n \in \omega \}$ be two disjoint sets of \emph{test variables} and \emph{program variables}, respectively. The set of \emph{KAT terms} is two sorted:
\begin{itemize}
\item tests: \qquad\quad $b, c := \mathsf{b}_n \mid \mathsf{0} \mid \mathsf{1} \mid \bar{b} \mid b \cdot c \mid b + c $
\item programs: \quad $p,q := \mathsf{p}_n \mid b \mid p + q \mid p \cdot q \mid p^{*}$
\end{itemize}
Note that each test is a program, and so we may use program variables to range over arbitrary KAT terms. The set of KAT terms is denoted as $\mathscr{T}_{\text{KAT}}$. If $\mathscr{B}$ is a KAT, then a \emph{$\mathscr{B}$-valuation} is any function $h$ from $\mathscr{T}_{\text{KAT}}$ to $\mathscr{B}$ such that
\begin{itemize}
\item $h(\mathsf{b}_n) \in B$
\item $h(\mathsf{0}) = 0$ and $h(\mathsf{1}) = 1$, and
\item $h$ commutes with $\star \in \{ \cdot, +, ^{*}, \, \bar{ } \, \}$.
\end{itemize}
A \emph{KAT equation} is an ordered pair of KAT terms, denoted in general as $p \eqt q$. An equation $p \eqt q$ \emph{holds} is a $\mathscr{B}$-valuation $h$ iff $h(p) = h(q)$. The \emph{equational theory of $\mathbf{KAT}$} is the set of KAT equations that holds in all $\mathscr{B}$-valuations where $\g{B} \in \mathbf{KAT}$. The equational theory of $\mathbf{KAT}$ ($\mathbf{rKAT}$) is denoted as $\mathsf{KAT}$ ($\mathsf{rKAT}$).

Propositional Hoare logic embeds into $\mathsf{KAT}$ \cite{Kozen2000},  $\mathsf{KAT}$ is PSPACE-complete \cite{CohenEtAl1996}, and $\mathsf{KAT} = \mathsf{rKAT}$ \cite{KozenSmith1997}.

\section{Kleene algebra with domain}\label{sec:KAD}

KAT is a two-sorted framework, and for various reasons it may be desirable to look for a one-sorted alternative. For instance, ``one-sorted domain semirings are easier to formalise in interactive proof assistants and apply in program verification and correctness.'' \cite[p.~576]{FahrenbergEtAl2022}. A one-sorted alternative called Kleene algebra with domain was introduced in \cite{DesharnaisStruth2011}, following earlier work on a related two-sorted framework \cite{DesharnaisEtAl2006}. The idea is to expand a Kleene algebra with an operation $d$ that transforms elements $x \in K$ into \emph{test-like} elements $d(x)$, and to show that the set of images of $x \in K$ under $d$, $d(K)$, forms a Boolean algebra. In other words, instead of assuming the existence of a Boolean subalgebra, such a subalgebra arises from images under a new operation $d$. 

The desired properties of $d$ are inspired by relational KAT, where tests are subsets of the identity relation. One obvious choice of a $d$ such that $d(R)$ is a subset of the identity relation is the \emph{relational domain} operation
\begin{equation*}
d(R) := \{ (s, s) \mid \exists u. (s, u) \in R \} 
\end{equation*}
Informally, $d(R)$ represents the set of states in which the program associated with $R$ has a terminating computation. The main question shaping the definition of KAD is the following: \emph{which properties of relational domain should be adopted on the abstract algebraic level in order to yield a Boolean subalgebra?}

\begin{definition}
A \emph{predomain} operation on a Kleene algebra is any $d : K \to K$ such that 
\begin{gather}
x \leq d(x)x\label{left_preserver}\\
d(xy) \leq d(xd(y))\label{sublocality}\\
d(x) \leq 1\\
d(0) = 0\\
d(x + y) = d(x) + d(y)
\end{gather}
A \emph{domain} operation on a Kleene algebra is any predomain operation that satisfies 
\begin{equation}\label{locality}
d(xy) = d(xd(y))
\end{equation}
\end{definition}

It is easily checked that the relational domain operation is a domain operation on relational Kleene algebras.

\begin{proposition}[\cite{DesharnaisStruth2011}]\label{prop:dKA}
Let $\mathscr{K}$ be a Kleene algebra and let $d$ be a domain operation on $\mathscr{K}$. Let $d(K) := \{ y \mid \exists x. y = d(x) \}$. Then:
\begin{enumerate}
\item $d(K)$ contains $1$ and $0$, and it is closed under $\cdot$ and $+$;
\item $d(\g{K}) = (d(K), \cdot, +, 1, 0)$ is a bounded distributive lattice.
\end{enumerate}
\end{proposition}

In order to obtain a Boolean algebra from the distributive lattice $d(\mathscr{K})$, one has to make sure that each test $d(x)$ is complemented in $d(K)$, that is, for each $d(x)$ there is $y \in d(K)$ such that $d(x)y = 0$ and $d(x) + y = 1$. An elegant solution to this problem presented in \cite{DesharnaisStruth2011} consists in expanding Kleene algebras with a single unary operation $a$ (\emph{antidomain}) that allows to define a domain operation $d$ and has properties entailing that $a(x)$ is a complement of $d(x)$.

\begin{definition}
A \emph{Kleene algebra with domain} is a Kleene algebra expanded with an operation $a : K \to K$ such that
\begin{gather}
a(x)x = 0\\
a(xy) \leq a(x\, a(a(y)) \label{antidomain:DNI}\\
a(x) + a(a(x)) = 1
\end{gather}
We define $d(x) := a(a(x))$.
\end{definition}

\begin{lemma}[\cite{DesharnaisStruth2011}]
In each KAD, $d$ is a domain operation and $a(x)$ is a complement of $d(x)$.
\end{lemma}

\begin{theorem}[\cite{DesharnaisStruth2011}]
In each KAD, $d(\g{K}) = (d(K), \cdot, +, 1, 0)$ is a Boolean algebra.
\end{theorem}

It has been shown that, at least in the semiring signature extended with $d$, the domain axioms capture the essential properties of relational domain \cite{McLean2020}. However, $d$ has some peculiar features when it comes to general Kleene algebra and KAT. First, $d(\mathscr{K})$ is necessarily the maximal Boolean subalgebra of the negative cone of $\mathscr{K}$; see Thm.~8.5 in \cite{DesharnaisStruth2011}. In a sense, then, every ``proposition'' is considered a test, contrary to some of the intuitions expressed in \cite{Kozen1997}. These intuitions also collide with the approach of taking KAT as KA with a Boolean negative cone \cite{FahrenbergEtAl2022,Jipsen2004}. Second, not every Kleene algebra expands to a KAD, not even every finite one; see Prop.~5.3 in \cite{DesharnaisStruth2011}; see the first counterexample in the appendix. This is in contrast to the fact that every Kleene algebra expands to a KAT. This feature is caused by \eqref{locality} and the authors of \cite{DesharnaisStruth2011} express interest in variants of $d$ not satisfying \eqref{locality}.\footnote{The immediate candidate is predomain, but it can be shown that if $d$ is a predomain operation then $d(K)$ is not necessarily closed under $\cdot$, and that $d(x) = d(x)d(x)$ may fail (see $\mathbf{A}_4$ in the appendix). 
} 

Hence, it is interesting to ask if there is a suitable one-sorted alternative to KAT, in the sense of satisfying the following properties:
\begin{itemize}
\item \textbf{(P0)} The algebras in the given class $\mathbf{ALT}$ expand Kleene algebras by additional operations $t$ and $t'$.

\item \textbf{(P1)} The test algebra $t(\mathscr{A})$ for each $\mathscr{A} \in \mathbf{ALT}$ is a Boolean algebra.

\item \textbf{(P2)} Every Kleene algebra expands to an $\mathscr{A} \in \mathbf{ALT}$.

\item \textbf{(P3)} The test algebra $t(\mathscr{A})$ is not necessarily the maximal Boolean subalgebra of the negative cone of $\mathscr{A}$.

\item \textbf{(P4)} The equational theory of $\mathbf{KAT}$ embeds into the equational theory of $\mathbf{ALT}$.
\end{itemize}

Condition (P2) would ensure that $\mathbf{ALT}$ is a \emph{conservative expansion} of Kleene algebras, i.e., $\mathbf{ALT}$ satisfies exactly the same $(\cdot, +, ^*, 1, 0)$-equations as Kleene algebras.

\section{One-sorted KAT}\label{sec:KAt}

In the present section, we define a class of algebras satisfying (P0--3). The algebras in the class will be called simply \emph{one-sorted Kleene algebras with tests}, KAt. Property (P4) is established in Section \ref{sec:embedding}.

\begin{definition}
A \emph{one-sorted Kleene algebra with tests} is a Kleene algebra expanded by two unary operations  $t$ and $t'$ such that
\begin{gather}
t(0) = 0\label{t:0}\\
t(1) = 1\label{t:1}\\
t(t(x) + t(y)) = t(x) + t(y)\label{t:weakly_additive}\\
t(t(x) t(y)) = t(x)\, t(y)\label{t:multi}\\
t(x)t(x) = t(x)\label{t:idem}\\
t(x) \leq 1\label{t:negcone}\\
1 \leq t'(t(x)) + t(x)\label{t:lem}\\
t'(t(x))\, t(x) \leq 0\label{t:ecq}\\
t'(t(x)) = t(t'(t(x)))\label{t:cc}
\end{gather}
\end{definition}

\begin{example}
Relational KA with $t(R) = \{ (w, w) \mid \exists v. \, (w, v) \in R) \}$ and $t'(R) = \{ (w, w) \mid \neg \exists v. \, (w, v) \in R) \}$.
\end{example}

Every KAD $\mathscr{D} = (K, \cdot, +, \,^{*}, 1, 0, a)$ can be seen as an algebra of the type of KAt, namely, by adding $d$ explicitly to the signature. It can then be shown that KAt generalizes KAD.

\begin{proposition}
Every KAD is a KAt.
\end{proposition}
\begin{proof}
It needs to be checked that the KAt axioms are satisfied by any KAD where $d$ is seen as $t$ and $a$ is seen as $t'$. This is mostly a matter of easy calculation using known KAD properties.
\end{proof}

Let $\mathscr{K} = (K, \cdot, +, ^{*}, 1, 0, t, t')$ be a KAt. Let $t(K) = \{ x \mid \exists y. x = t(y) \}$. 

\begin{lemma}\label{lem:tK-welldef}
For all KAt, $d(K)$ is closed under the semiring operations and $t'$.
\end{lemma}
\begin{proof}
The set $d(K)$ contains constants $1, 0$ by (\ref{t:0}--\ref{t:1}), and it is closed under $+$ and $\cdot$ by (\ref{t:weakly_additive}--\ref{t:multi}). $d(K)$ is closed under $t'$ by \eqref{t:cc}.  
\end{proof}

For each KAt $\mathscr{K}$, let $t(\mathscr{K}) = (t(K), \cdot, +, 1, 0, t')$. This is a well defined subalgebra of $\mathscr{K}$ by Lemma \ref{lem:tK-welldef}.

\begin{theorem}\label{thm:tK-BA}
Each $t(\mathscr{K})$ is a Boolean algebra. 
\end{theorem}
\begin{proof}
Each $t(\mathscr{K})$ is a bounded distributive lattice by (\ref{t:idem}--\ref{t:negcone}) and properties of semirings. By (\ref{t:lem}--\ref{t:ecq}), $t'(x)$ is a complement of $x$ for all $x \in t(K)$. Since $t(\mathscr{K})$ is a bounded distributive lattice, it is uniquely complemented, and so it is a Boolean algebra.
\end{proof}

\begin{theorem}\label{thm:KAexpands_toKAt}
Each Kleene algebra expands to a KAt.
\end{theorem}
\begin{proof}
The proof is similar as in the case of KAT. For any Kleene algebra $(K, \cdot, +, ^{*}, 1, 0)$, define:
\begin{center}
$t(x) = \begin{cases}
0 & \text{if } x = 0\\
1 & \text{otherwise.}
\end{cases}$ \qquad
$t'(x) =  \begin{cases}
0 & \text{if } x = 1\\
1 & \text{if } x = 0\\
x & \text{otherwise.}
\end{cases}$
\end{center}
The rest is checked by easy calculation.
\end{proof}

Theorem \ref{thm:KAexpands_toKAt} shows that, unlike KAD and like KAT, KAt is a proper generalization of Kleene algebra. The proof of the theorem also shows that $t(\mathscr{K})$ need not be the maximal Boolean subalgebra of the negative cone.

On the other hand, one may say that $t(x)$ lacks any intuitive relation to $x$, unlike the domain operator in KAD. For instance, it holds in KAD that $d(x)$ is a left preserver of $x$, that is, $x \leq d(x)x$. This lack of connection between $x$ and $t(x)$ can be seen as the reason why (P0--3) are satisfied. However, the following section shows that this is only partially so.

\section{An extension of KAt}\label{sec:ext}

In this section we consider a particular extension of KAt with a number of axioms known from KAD. We show that the extension still has properties (P0--3). 

\begin{definition}
A \emph{strong KAt}, sKAt, is a KAt satisfying the following axioms:
\begin{gather}
t(x +y) = t(x) + t(y) \label{t:additivity}\\
x \leq t(x) x \label{t:left_preserver}\\
t (t(x)y) \leq t(x) \label{t:least_left_preserver}\\
t(xy) \leq t(xt(y)) \label{t:sublocality}
\end{gather}
\end{definition}

\eqref{t:additivity} entails that $t$ is monotonic; \eqref{t:left_preserver} says that $t(x)$ is a left preserver of $x$; \eqref{t:least_left_preserver} entails
\begin{equation}
 x \leq t(y) x \implies t(x) \leq t(y) \, ,
 \end{equation} 
 saying that $t(x)$ is the \emph{least} left preserver among tests; and \eqref{t:sublocality} is the sublocality property of predomain operators. In fact, $t$ is a predomain operator in each sKAt. 
The algebra from Prop.~5.3 in \cite{DesharnaisStruth2011} shows that $t$ need not be a domain operation; see the first example in the appendix.
 
 It is clear that sKAt satisfies (P0) and (P1). It can also be shown, by revisiting the construction in the proof of Theorem \ref{thm:KAexpands_toKAt} that (P2) and (P3) are satisfied as well. However, as shown already in \cite{DesharnaisStruth2011}, adding locality $t(xt(y)) \leq t(xy)$ destroys (P2) and (P3).
 
 \section{Embedding KAT}\label{sec:embedding}
 
 In this section we show that KAt has property (P4) as well. The result actually applies to a larger class of cases.
 
We say that a KAT $\mathscr{K} = (K, B, \cdot, +, ^{*}, 1, 0, \, \bar{ } \, )$ \emph{expands} into a KAt $\mathscr{A} = (K, \cdot, +, \,^{*}, 1, 0, t, t')$ iff
 \begin{itemize}
 \item $B = t(K)$.
\end{itemize} 
Note that if $\mathscr{K}$ expands to $\mathscr{A}$, then $\bar{b} = t'(b)$, since $\bar{b} = \overline{t(y)} = t't(x) = t'(b)$. The second to last equation holds since $t't(x)$ is the complement of $t(x)$ in $t(K) = B$.
 
 \begin{definition}
 A class of KAt $\mathbf{A}$ is \emph{KAT-like} iff it satisfies properties (P0), (P1) and
 \begin{itemize}
 \item {\normalfont \textbf{(P2a)}} Every KAT expands to a $\mathscr{A} \in \mathbf{A}$. 
 \end{itemize}
 \end{definition}
 
 \begin{theorem}
 If $\mathbf{A}$ is KAT-like, then $\mathbf{A}$ satisfies (P4).
 \end{theorem}
 
 In the proof of the theorem, we will use the following translation from $\mathscr{T}_{\text{KAT}}$ to $\mathscr{T}_{\text{KAt}}$:
 \begin{itemize}
\item $\tra (\at{p}_n) = \at{x}_{2n}$
\item $\tra (\at{b}_n) = t( \at{x}_{2n +1} )$
\item $\tra (\mathsf{0}) = t(\mathsf{0})$
\item $\tra (\mathsf{1}) = t(\mathsf{1})$
\item $\tra (p + q) = \tra (p) + \tra (q)$
\item $\tra (p \cdot q) = \tra (p) \cdot \tra (q)$
\item $\tra (p^{*}) = \tra (p)^{*}$
\item $\tra (\bar{b}) = t' (\tra (b))$
\end{itemize}

\begin{proof}
We show that, for all KAT equations $p \eqt q$, there is KAT $\mathscr{K} \not\models p \eqt q$ iff there is $\mathscr{A} \not\models \tra (p) \eqt \tra(q)$ such that $\mathscr{A} \in \mathbf{A}$.

First, if $\mathscr{K} \not \models p \eqt q$, then there is a $\mathscr{K}$-valuation $\val{\,}$ such that $\val{p} \neq \val{q}$. By (P2a), $\mathscr{K}$ expands to $\mathscr{A} \in \mathbf{A}$. We define an $\mathscr{A}$-valuation $\vval{\,}$ as the unique homomorphism such that
\begin{itemize}
\item $\vval{\mathsf{x}_{2n}} = \val{\mathsf{p}_n}$,
\item $\vval{\mathsf{x}_{2n+1}} = \val{\mathsf{b}_n}$.
\end{itemize}
We prove by induction on KAT terms $r$ that
\begin{equation}\label{val-equivalence}
\val{r} = \vval{\tra (r)}
\end{equation}
Before the main inductive proof, we prove that \eqref{val-equivalence} holds for all tests $b$; the proof is by induction on tests. The base case is established as follows: $\vval{\tra (\mathsf{b}_{n}) } = \vval{t (\mathsf{x}_{2n+1})} = t \vval{ \mathsf{x}_{2n+1}} = t \val{\mathsf{b}_{n}} = \val{\mathsf{b}_{n}}$. The last equation holds by definition of expansion (we know that $\val{\mathsf{b}_{n}} \in B = t(K)$ and that $tt(x) = t(x)$). The cases of the induction step for $+$ and $\cdot$ are trivial. The case for $\, \bar{ } \, $ is established as follows: $\vval{\tra (bar{b})} = \vval{ t' \tra (b)} = t' \vval{\tra(b)} = t' \val{b} = \val{ t'(b) } = \val{ \bar{b}}$ by the definition of expansion (see the note after the definition). This concludes the proof for tests. The main inductive proof for KAT terms is now trivial.

The converse claim is established as follows. Assume that there is $\mathscr{A} \in \mathbf{A}$ such that $\mathscr{A} \not\models \tra (p) \eqt \tra (q)$. Hence, there is an $\mathscr{A}$-valuation $\vval{ \,}$ such that $\vval{\tra (p)} \neq \vval{\tra (q)}$.  It follows from (P1) that $t(\mathscr{A})$ is a Boolean algebra, and so $\mathscr{K} = (K, t(\mathscr{A}), \cdot, +, \, ^{*}, 1, 0, t')$ is a KAT. We define a $\mathscr{K}$-valuation $\val{\,}$ as the unique KAT-homomorphism such that
\begin{itemize}
\item $\val{\mathsf{p}_{n}} = \vval{\mathsf{x}_{2n}}$
\item $\val{\mathsf{b}_{n}} = \vval{ t (\mathsf{x}_{2n+1})}$
\end{itemize}
As before, it is sufficient to prove that \eqref{val-equivalence} holds in this setting for all KAT terms $r$. The proof is again by induction and the only interesting case is $\bar{b}$ (the case for $\mathsf{b}_n$ holds by definition). We reason as follows: $\vval{\tra (\bar{b})} = \vval{ t' \tra(b)} = t' \val{b} = \val{ t' (b) } = \val{\bar{b}}$. 
\end{proof}

\begin{theorem}
$\mathbf{KAt}$ is KAT-like.
\end{theorem}
\begin{proof}
Take an arbitrary KAT $\mathscr{K} = (K, B, \cdot, +, \,^{*}, 1, 0, \, \bar{ }\, )$ and define $\mathscr{A} = (K, \cdot, +, \,^{*}, 1, 0, t, t')$ where:
\begin{center}
$t(x) = \begin{cases}
x & \text{if } x \in B\\
1 & \text{otherwise.}
\end{cases}$ \qquad
$t'(x) =  \begin{cases}
\bar{x} & \text{if } x \in B\\
x & \text{otherwise.}
\end{cases}$
\end{center}
It can be easily checked that $\mathscr{A}$ is a KAt.
\end{proof}

It is a matter of easy calculation that $\mathscr{A}$ in the proof of the above theorem satisfies \eqref{t:left_preserver} and \eqref{t:least_left_preserver}. For \eqref{t:additivity} and \eqref{t:sublocality}, a different approach is needed.

It is known that the equational theory of KAT is identical to the equational theory of relational KAT \cite{KozenSmith1997}. Hence, one could work with a modification of the notion of a KAT-like class of algebras.

\begin{definition}
A class of KAt $\mathbf{A}$ is \emph{rKAT-like} iff it satisfies properties (P0), (P1) and
 \begin{itemize}
 \item {\normalfont \textbf{(P2b)}} Every relational KAT expands to a $\mathscr{A} \in \mathbf{A}$. 
 \end{itemize}
\end{definition}

\begin{theorem}\label{thm:embedding}
 If $\mathbf{A}$ is rKAT-like, then $\mathbf{A}$ satisfies (P4).
\end{theorem}
\begin{proof}
The fact that \eqref{val-equivalence} holds in both directions of the proof is established as before.
\end{proof}

\begin{theorem}
$\mathbf{sKAt}$ and $\mathbf{KAD}$ are rKAT-like.
\end{theorem}
\begin{proof}
Take any relational KAT, and define 
\begin{itemize}
\item $t(R) = \{ (s,s) \mid \exists u. (s, u) \in R \}$;
\item $t'(R) = \{ (s,s) \mid \neg \exists u. (s, u) \in R \}$.
\end{itemize}
It is easily checked that we obtain a sKAt.
\end{proof}

Since it is known that Propositional Hoare Logic embeds into KAT \cite{Kozen2000}, we obtain the result that KAt, sKAt and KAD subsume PHL as well (the latter case is known).
 
 \section{Residual}\label{sec:residual}
 
 In this section we focus on a particular version of KAt where $t'$ is not a primitive operator, but it is defined using the residual of Kleene algebra multiplication. It is known Kleene algebras with residuals form a variety \cite{Pratt1991}. Consequently, one-sorted residuated Kleene algebras and their equational extensions form varieties as well.

\begin{definition}
A \emph{residuated Kleene algebra} is a Kleene algebra expanded with a binary operation $\to$ such that
\begin{equation}
xy \leq z \iff x \leq y \to z
\end{equation}
\end{definition} 

\begin{definition}
A \emph{residuated one-sorted Kleene algebra with tests} is a residuated Kleene algebra expanded with a unary operation $t$ satisfying (\ref{t:0}--\ref{t:negcone}) and
\begin{gather}
1 \leq t(t(x) \to 0) + t(x)\label{r:lem}\\
t(t(x) \to 0) \, t(x) \leq 0\label{r:ecq}
\end{gather}
We define $t'(x) = t(t(x) \to 0)$.
\end{definition}

It is an easy consequence of the definition that the class of residuated KAt forms a (finitely based) variety. We denote it as $\rKAt$.

\begin{proposition}
Each residuated KAt is a KAt.
\end{proposition}
\begin{proof}
It is easy to check that $t'$ as defined in terms of $t$ and $\to$ satisfies (\ref{t:lem}--\ref{t:ecq}). It follows from  \eqref{t:1} and \eqref{t:multi} that $t(t(x)) = t(x)$. Hence, \eqref{t:cc} holds as well.  
\end{proof}

\begin{theorem}\label{thm:KAexpands_toResKAt}
Each join-complete Kleene algebra expands to a residuated KAt.
\end{theorem}
\begin{proof}
Each join-complete KA is residuated since we can define
$$ x \to y := \sum \{ z \mid zx \leq y \} \, .$$ Define $t(x)$ as in the proof of Theorem \ref{thm:KAexpands_toKAt}. We prove by cases that (\ref{r:lem}--\ref{r:ecq}) hold.
\begin{itemize}
\item If $x = 0$: $t(t(x) \to 0) + t(x) = t(0 \to 0) + 0 = t(1) + 0 = 1$ and $t(t(x) \to 0) \, t(x) = 10 = 0$;
\item If $x \neq 0$: $t(t(x) \to 0) + t(x) = t( 1 \to 0) + 1 = t(0) + 1 = 1$ and $t(t(x) \to 0) \, t(x) = 01 = 0$.\qedhere 
\end{itemize}
\end{proof}

Not every KA expands to a residuated KA \cite{Kozen1994a}, so not every KA expands to a residuated KAt.

\section{Conclusion}

In this paper we studied a one-sorted variant of Kleene algebra with tests, KAt. We have shown that KAt generalizes Kleene algebra with domain, KAD, and that it retains the good properties of KAD, namely (P1) the test algebra of any KAt is a Boolean algebra and (P4) the equational theory of KAT embeds into the equational theory of KAt. Moreover, KAt avoids the bad properties of KAD: unlike KAD, every Kleene algebra expands to a KAt (P2), and the test algebra in KAt is not necessarily the maximal Boolean subalgebra of the negative cone (P3). We have shown that adding several axioms familiar from KAD to KAt does in fact not injure (P2) and (P3). We believe that this leads to a better general view of the landscape of one-sorted program algebras.


\appendix

\section{Counter-examples}\label{sec:counterexamples}
The following algebra $\mathscr{A}_3$ exhibits a Kleene algebra with predomain operation $d$, which does not satisfy the locality axiom $d(2\cdot 2)=0\neq 1=d(2\cdot d(2))$. This example is closely related to the example in Prop.~5.3 in \cite{DesharnaisStruth2011}.
\begin{center}
\begin{tabular}{cccc}
\begin{picture}(30,30)(0,5)

\put(12,24){\circle*{4}}
\put(16,20){$1$}
\put(12,12){\line(0,1){12}}
\put(12,12){\circle*{4}}
\put(16,8){$2$}
\put(12,0){\line(0,1){12}}
\put(12,0){\circle*{4}}
\put(16,-4){$0$}
\put(7,-15){$\mathscr{A}_3$}

\end{picture}
&
\begin{tabular}{r|rrr}
$\cdot$ & 0 & 1 & 2\\
\hline
    0 & 0 & 0 & 0 \\
    1 & 0 & 1 & 2 \\
    2 & 0 & 2 & 0
\end{tabular}
&
\begin{tabular}{r|rrr}
$^*$ & 0 & 1 & 2\\
\hline
   & 1 & 1 & 1
\end{tabular}
&
\begin{tabular}{r|rrr}
$d$ & 0 & 1 & 2\\
\hline
   & 0 & 1 & 1
\end{tabular}

\end{tabular}
\end{center}
The following algebra $\mathscr{A}_4$ exhibits a Kleene algebra with predomain operation $d$, such that $d(A_4)$ is not closed under $\cdot$, since $d(2)=3\neq 2=d(2)d(2)\notin d(A_4)$.
\begin{center}
\begin{tabular}{cccc}
\begin{picture}(30,40)(0,5)

\put(12,36){\circle*{4}}
\put(16,32){$1$}
\put(12,24){\line(0,1){12}}
\put(12,24){\circle*{4}}
\put(16,20){$3$}
\put(12,12){\line(0,1){12}}
\put(12,12){\circle*{4}}
\put(16,8){$2$}
\put(12,0){\line(0,1){12}}
\put(12,0){\circle*{4}}
\put(16,-4){$0$}
\put(7,-15){$\mathscr{A}_4$}

\end{picture}
&
\begin{tabular}{r|rrrr}
$\cdot$ & 0 & 1 & 2 & 3\\
\hline
    0 & 0 & 0 & 0 & 0 \\
    1 & 0 & 1 & 2 & 3 \\
    2 & 0 & 2 & 0 & 0 \\
    3 & 0 & 3 & 2 & 3
\end{tabular}
&
\begin{tabular}{r|rrrr}
$^*$ & 0 & 1 & 2 & 3\\
\hline
   & 1 & 1 & 1 & 1
\end{tabular}
&
\begin{tabular}{r|rrrr}
$d$ & 0 & 1 & 2 & 3\\
\hline
   & 0 & 1 & 3 & 1
\end{tabular}

\end{tabular}
\end{center}
\end{document}